\documentclass[12pt]{article}
\usepackage{amsmath}
\usepackage{amsthm}
\usepackage{amssymb}
\newtheorem{proposition}{Proposition}
\newtheorem{theorem}{Theorem}
\newtheorem{corollary}{Corollary}
\newtheorem{example}{Example}[section]
\usepackage{graphicx}
\usepackage{enumerate}
\usepackage{url} 
\usepackage{setspace}
\raggedright

\newcommand{\blind}{0}

\begin{document}

\def\spacingset#1{\renewcommand{\baselinestretch}%
{#1}\small\normalsize} \spacingset{1}


\if0\blind
{
  \title{\bf The Exponentiated Hypoexponential Disribution.}
  \author{Anass Nassabein\thanks{
    The authors gratefully acknowledge \textit{please remember to list all relevant funding sources in the unblinded version}}\hspace{.2cm}\\
    Department of Mathematics, Lebanese International University\\
    Therrar Kadri \\
    Department of Science, Northwestrern Polytecgnic\\
    Seifideen Kadry \\
    Department of Applied Computing and Technology, Noroff University College\\
    Khaled Smaili \\
    Department of Mathematics, Lebanese University}
 
 \maketitle
  } \fi
\if0\blind
{
  \bigskip
  \bigskip
  \bigskip
  \begin{center}
    {\LARGE\bf The Exponentiated Hypoexponential Disribution}
\end{center}
  \medskip
} \fi

\bigskip
\begin{abstract}
In this paper we study the Exponentiated Hypoexponential Distribution with
different parameters. The distribution added a parameter to the $n$
parameters of the Hypoexponenial distribution. We first derive a closed
expression of the probability density function and the cumulative
distribution function of Maximum Exponentiated Exponential distribution.
These functions are used to determine an exact expression of the probability
density function, the cumulative function, reliability function, and hazard
function of our Exponentiated Hypoexponential Distribution. We discuss
estimation of the parameters by maximum likelihood estimators. The
distribution has been fitted to a real life data set and the fit has been
found to be a serious competitor to the others.
\end{abstract}

\noindent%
{\it Keywords:} Maximum Exponentiated Exponential Distribution, Exponentiated Hypoexponential Distribution, Maximum Likelihood Estimation.
\vfill

\newpage
\spacingset{1.9} 
\section{Introduction}
\label{sec:intro}

The Exponential random variable like its alternatives, Erlang and Gamma,
plays a significant role in both queueing and reliability problems \cite{ab}
and \cite{q}. The time between independent events with constant average or
time to failure are typical examples of the Exponential random variable.
Some examples in many fields of science are: how long a cell-phone call
lasts \cite{fa}, how long it takes a computer network to transmit a message
from one node to another, and the distance between mutations on a DNA
stand..., see, \cite{Phd}.

In addition to that, the function of random variable possesses an important
role in the event modeling in many applied problems of engineering, physics,
economics, biology, genetics, medicine, see, \cite{b} and \cite{Phd}. One of
the important functions of random variable is the sum of independent random
variables especially the sum of independent Exponential distributions which
is known to be the Hypoexponential distribution. The Hypoexponential
distribution that also plays a significant role in modeling many events in
most domains of science, Trivedi \cite{25}, Markov process, Kadri et al. 
\cite{44}.

There are many ways to add one or more parameters to a parent distribution
function. Adding the parameters makes the resulting distribution more
flexible and richer for modeling data. Marshall and Olkin \cite{1997} added
a positive parameter to a general reliability function (RF). In their
consideration of a countable mixture of Pascal $(r,p)$ mixing proportion and
positive integer powers of RF, they obtained a reliability function with two
extra parameters in AL-Hussaini and Ghitany \cite{2005}. AL-Hussaini and
Gharib \cite{2009} obtained a new family of distributions as a countable
mixture with Poisson by adding a parameter.

The idea of Exponentiated distribution was introduced by Gupta et al. \cite%
{1998}, who discussed a new family of distributions called Exponentiated
Exponential distribution. The family has two parameters similar to that of
Weibull or Gamma's family. Gupta and Kundu studied some properties of the
distribution see, \cite{2001a}. They concluded that this distribution can be
an alternative possibility to Gamma or Weibull family, because their
observation was clear that several properties of the new family are similar
to those of the Weibull or Gamma family. They also examined the estimation
and inference aspects of the distribution see, \cite{2001b}, and \cite{2002}%
. The distribution has been further studied by Nadarajah and Kotz \cite{2003}%
. A class of goodness-of-fit tests for the distribution with estimated
parameter has been proposed by Hassan \cite{2005b}. Pal et al. \cite{2006}
studied the Exponentiated Weibull family as an extension of the Weibull and
Exponentiated Exponential families.

In this paper, we introduce the Exponentiated Hypoexponential distribution.
We start by characterizing the related functions that well describe our new
distribution. We show it by relating it to the Maximum Exponentiated
Exponential distribution. Next, we give an estimate of the new parameter k,
by MLE method. Eventually we run a simulation of data showing the
EHypoexponential distribution fits better than the original distribution.
The problem is solved completely by considering the different parameters of
Hypoexponential random variable.

\section{Preliminaries}

\label{EE Distribution}The Exponentiated Exponential (EE) distribution, and
known as the generalized Exponential distribution, is a particular case of
Gompertz-Verhulst distribution function when the parameter $\rho =1$. Thus, $%
X\sim EE\left( \lambda _{i},\alpha _{i}\right) $ is a two- parameter
Exponentiated Exponential random variable when it has the following
comulative distribution function (CDF);%
\begin{equation}
F(x;\alpha ;\lambda )=\left( 1-e^{-\lambda x}\right) ^{\alpha };\text{ }x>0
\label{CDF EE}
\end{equation}%
and therefore its probability density function (PDF) is of the form%
\begin{equation}
f(x;\alpha ;\lambda )=\alpha \lambda \left( 1-e^{-\lambda x}\right) ^{\alpha
-1}e^{-\lambda x};x>0  \label{PDF EE}
\end{equation}

\subsection{The Hypoexponential Distribution}

The Hypoexponential random variable is the sum of independent Exponential
random variable. This distribution is examined accoding to the stages of the
Exponential random variable.\ Whenever the stages are identical the
distribution obtained is Erlang distribution. The other case of the
Hypoexponential distribution is when the stages are different or the
parameters are exponential distribution are different. This case was
examined completlet by Samili et al. in \cite{Hy.diff}. They denoted this
distribution as the Hypoexpoential distribution with different parameters
and given as follows:

Let $X_{1\text{ }},X_{2\text{ }},...,X_{n\text{ }}$be independent
Exponential random variables with different, respective parameters $\alpha
_{i},$ $i=1,2,...,n$. Thus $S_{n}=\sum\limits_{i=1}^{n}X_{i\text{ }}$is the
Hypoexponential random variable with parameters and written as $S_{n}\sim
hypoexp(\overrightarrow{\alpha })$, where $\overrightarrow{\alpha }%
_{i}=\left( \alpha _{1},\alpha _{2},...,\alpha _{n}\right) $. Samili et al.
in \cite{Hy.diff} showed that the density function of $S_{n}$ is given as 
\begin{equation*}
f_{S_{n}}(t)=\sum\limits_{i=1}^{n}A_{i}f_{X_{i\text{ }}}(t)
\end{equation*}%
where 
\begin{equation}
A_{i} = \prod\limits_{j=1,\,j\neq i}^{n} \left(\frac{\alpha_{j}}{\alpha_{j}-\alpha_{i}}\right)
\label{EqAi}
\end{equation}
is denoted as the coefficients of Hypoexponential. The author as a
consequence detemined the other related function to $S_{n}$ as the CDF as 
\begin{equation}
F_{S_{n}}(t)=\sum\limits_{i=1}^{n}A_{i}F_{X_{i\text{ }}}(t)  \label{Hyp CDF}
\end{equation}

\subsection{Exponentiated Distributions}

Let X be a random variable with probability density function $\left(
PDF\right) $ $f\left( x\right) $ and the cumulative distributions function $%
\left( CDF\right) $\ $F\left( x\right) $, $x\in R^{1}$. Consider a random
variable $Z$ with the $CDF$ given by

\begin{equation}
G_{\alpha }\left( z\right) =\left[ F\left( z\right) \right] ^{\alpha },\text{
}z\in R^{1},\alpha >0  \label{CDF ExpoRV}
\end{equation}

Then Z is said to have an exponentiated distribution.

The $pdf$ of $Z$ is given by

\begin{equation}
g_{\alpha }\left( z\right) =\alpha \left[ F\left( z\right) \right] ^{\alpha
-1}f\left( z\right)  \label{PDF ExpoRV}
\end{equation}

The survival function $R_{\alpha }\left( z\right) $ and the failure rate $%
h_{\alpha }\left( z\right) $ of the distribution of Z are defined by

\begin{equation*}
R_{\alpha }\left( z\right) -1-G_{\alpha }\left( z\right)
\end{equation*}

\begin{equation*}
h_{\alpha }\left( z\right) =\frac{g_{\alpha }\left( z\right) }{1-G_{\alpha
}\left( z\right) }
\end{equation*}

\section{Maxmumim Exponentiated Exponential Distribution}

In this section we find a simplified expressions of PDF and CDF of the
Maximum Exponentiated Exponential distribution. The expressions obtained are
used later to derive the function of our new distribution the
EHypoexponential distribution.

In the following theorem we find a simplified expression of PDF and CDF of
the Maximum Exponentiated Exponential distribution when the exponents of the
EE distribution is a positive real number.

\begin{theorem}
\label{thcdfMEE1} Let $X_{i}\sim EE(\lambda _{i},\alpha _{i})\ $where $\alpha
_{i}\in 
\mathbb{R}
_{+}^{\ast },$ $\lambda _{i}>0,$ and $\lambda _{i}\neq \lambda _{j}$ for all 
$i\neq j,i=1,2,...,n$ and let $N=\max \left\{ X_{1},X_{2},...,X_{n}\right\} $
denoted as $N\sim MEE(\overrightarrow{\lambda },\overrightarrow{\alpha }),$ $%
\overrightarrow{\lambda }=\left( \lambda _{1},\lambda _{2},...,\lambda
_{n}\right) ,$ $\overrightarrow{\alpha }=\left( \alpha _{1},\alpha
_{2},...,\alpha _{n}\right) .$ Then the CDF of $N$ is 
\begin{equation*}
F_{N}(t) = \prod_{j=1}^{n} \left( \sum_{i=0}^{\infty} \binom{\alpha_{j}}{i} (-1)^{(\alpha_{j}-i)} e^{-\lambda_{j}(\alpha_{j}-i)t} \right)
\end{equation*}

and the PDF

\begin{equation*}
f_{N}(t) = \left( \prod_{j=1}^{n} (1-e^{-\lambda_{j}t})^{\alpha_{j}} \right) \left( \sum_{j=1}^{n} \frac{\lambda_{j}\alpha_{j}e^{-\lambda_{j}t}}{1-e^{-\lambda_{j}t}} \right)
\end{equation*}
\begin{proof}
Let $X_{i}$ be an Exponentiated Exponential distribution with rate parameter $\lambda_{i}$ and exponent parameter $\alpha_{i}$, $i=1,2,...,n$, where $\lambda_{i} \neq \lambda_{j}$ for all $i \neq j$. Then the CDF of the Maximum Exponentiated Exponential distribution is given by
\begin{eqnarray*}
F_{N}(t) &=& F_{X_{1}}(t) F_{X_{2}}(t) ... F_{X_{n}}(t) \\
&=& (1-e^{-\lambda_{1}t})^{\alpha_{1}}(1-e^{-\lambda_{2}t})^{\alpha_{2}}...(1-e^{-\lambda_{n}t})^{\alpha_{n}} \\
&=& \prod_{j=1}^{n}(1-e^{-\lambda_{j}t})^{\alpha_{j}}
\end{eqnarray*}
\end{proof}
\end{theorem}
\begin{corollary}
Let $X_{i} \sim EE(\lambda_{i}, g_{i})$ where $g_{i} \in \mathbb{N}^{\ast}$, $\lambda_{i} > 0$, and $\lambda_{i} \neq \lambda_{j}$ for all $i \neq j, i = 1,2,...,n$. Let $N = \max\{X_{1}, X_{2},...,X_{n}\}$ denoted as $N \sim MEE(n, \overrightarrow{\lambda}, \overrightarrow{g})$, $\overrightarrow{\lambda} = (\lambda_{1}, \lambda_{2},...,\lambda_{n})$, $\overrightarrow{g} = (g_{1}, g_{2},...,g_{n})$. Then the CDF of $N$ is
\begin{equation*}
F_{N}(t) = \prod_{j=1}^{n}\left(\sum_{i=0}^{\infty}\binom{\alpha_{_{j}}}{i}(-1)^{(\alpha_{_{j}}-i)}e^{-\lambda_{j}(\alpha_{_{j}}-i)t}\right)
\end{equation*}
where $\binom{\alpha_{_{j}}}{i}=\frac{(\alpha)_{j}}{i!}=\frac{\alpha(\alpha-1)..\left(\alpha-j+1\right)}{i!}$ is the generalized binomial coefficient. The PDF of $N$ is given by
\begin{equation*}
f_{N}(t) = \left(\prod_{j=1}^{n}(1-e^{-\lambda_{j}t})^{\alpha_{j}}\right)\left(\sum_{j=1}^{n}\frac{\lambda_{j}\alpha_{j}e^{-\lambda_{j}t}}{1-e^{-\lambda_{j}t}}\right)
\end{equation*}
\end{corollary}
Next, we consider the same problem by find a simplified expression of PDF
and CDF of the Maximum Exponentiated Exponential distribution but the
exponents of the EE distribution is a positive integer.

\begin{corollary}
\label{thcdfMEE}
Let $X_{i} \sim EE(\lambda_{i}, g_{i})$ where $g_{i} \in \mathbb{N}^{\ast}$, $\lambda_{i} > 0$, and $\lambda_{i} \neq \lambda_{j}$ for all $i \neq j, i = 1,2,...,n$. Let $N = \max\{X_{1}, X_{2},...,X_{n}\}$ denoted as $N \sim MEE(n, \overrightarrow{\lambda}, \overrightarrow{g})$, $\overrightarrow{\lambda} = (\lambda_{1}, \lambda_{2},...,\lambda_{n})$, $\overrightarrow{g} = (g_{1}, g_{2},...,g_{n})$. Then the CDF of $N$ is
\begin{equation}
F_{N}(t) = \prod_{j=1}^{n}\left(\sum_{i=0}^{g_{_{j}}}\binom{g_{_{j}}}{i}(-1)^{(g_{_{j}}-i)}e^{-\lambda_{j}(g_{_{j}}-i)t}\right)
\label{CDF MEExp}
\end{equation}
and the PDF is
\begin{equation}
f_{N}(t) = \left(\prod_{j=1}^{n}(1-e^{-\lambda_{j}t})^{g_{j}}\right)\left(\sum_{j=1}^{n}\frac{\lambda_{j}g_{j}e^{-\lambda_{j}t}}{1-e^{-\lambda_{j}t}}\right)
\label{PDF MEExp}
\end{equation}
\end{corollary}
\begin{proof}
The steps of the proof are similar to the proof of Theorem \ref{thcdfMEE1}
but we use instead of the generalized binomial expansion, the usual binomial
expansion where 
\begin{equation*}
(1-e^{-\lambda _{j}t})^{g_{j}}=\sum_{i=0}^{g_{_{j}}}\binom{g_{_{j}}}{i}%
\left( -1\right) ^{(g_{_{j}}-i)}e^{-\lambda _{j}(g_{_{j}}-i)t}.
\end{equation*}
\end{proof}

In the following we give the reliability and hazard function of the Maximum
Exponentiated Exponential distribution in the case when the exponents of the
EE distribution is positive real number and positive integer.

We have $R_{N}(t)=1-F_{N}\left( t\right) $, thus when the exponents of the
EE distribution is positive real number, from Theorem \ref{thcdfMEE1} we get 
\begin{equation*}
R_{N}(t) = 1-\prod_{j=1}^{n}\left(\sum_{i=0}^{\alpha_{j}}\binom{\alpha_{j}}{i}(-1)^{(\alpha_{j}-i)}e^{-\lambda_{j}(\alpha_{j}-i)t}\right)
\end{equation*}
Also when the exponents of the EE distribution is positive integer, from
Corollary \ref{thcdfMEE}, we get 
\begin{equation}
R_{N}(t) = 1-\prod_{j=1}^{n}\left(\sum_{i=0}^{g_{j}}\binom{g_{j}}{i}(-1)^{(g_{j}-i)}e^{-\lambda_{j}(g_{j}-i)t}\right)
\label{Rela MEE}
\end{equation}
and we have $h_{N}(t)=\frac{f_{N}\left( t\right) }{R_{N}\left( t\right) }$,
thus when the exponents of the EE distribution are positive real number,
from Theorem \ref{thcdfMEE1} we get 
\begin{equation*}
h_{N}(t) = \frac{\left(\prod_{j=1}^{n}(1-e^{-\lambda_{j}t})^{\alpha_{j}}\right)\left(\sum_{j=1}^{n}\frac{\lambda_{j}\alpha_{j}e^{-\lambda_{j}t}}{1-e^{-\lambda_{j}t}}\right)}{1-\prod_{j=1}^{n}\left(\sum_{i=0}^{\infty}\binom{\alpha_{j}}{i}(-1)^{(\alpha_{j}-i)}e^{-\lambda_{j}(\alpha_{j}-i)t}\right)}
\end{equation*}
Also when the exponents of the EE distribution is positive integer, from
Corollary \ref{thcdfMEE}, we get 
\begin{equation}
h_{N}(t) = \frac{\left(\prod_{j=1}^{n}(1-e^{-\lambda_{j}t})^{g_{j}}\right)\left(\sum_{j=1}^{n}\frac{\lambda_{j}g_{j}e^{-\lambda_{j}t}}{1-e^{-\lambda_{j}t}}\right)}{1-\prod_{j=1}^{n}\left(\sum_{i=0}^{g_{j}}\binom{g_{j}}{i}(-1)^{(g_{j}-i)}e^{-\lambda_{j}(g_{j}-i)t}\right)}
\label{Hazard MEE}
\end{equation}
\section{Exponentiated Hypoexponential Distribution}

Let $S_{n}\sim hypoexp(\overrightarrow{\alpha })\ $where $n\in 
\mathbb{N}
^{\ast }$ and $\alpha _{i}>0,$ $i=1,2,...,n$ and let $Y$ be the
exponentiated distribution of $S_{n}$, with exponent $k$. We name $Y$ as the
Exponentiated Hypoexponential distribution or EHypoexponential with exponent 
$k$ and denoted as $Y\sim EHypo\exp \left( \overrightarrow{\alpha },k\right)
,\overrightarrow{\alpha }=\left( \alpha _{1},\alpha _{2},...,\alpha
_{n}\right) $ and $k\in 
\mathbb{N}
^{\ast }$. Thus Exponentiated Hypoexponential distribution has $n+1$
parameters $\left( \alpha _{1},\alpha _{2},...,\alpha _{n},k\right) $. We
aim to examine this new distribution, and find simple and closed expressions
of some statistical functions. We start by relating the CDF of the
Exponentiated Hypoexponential distribution to CDF of the maximum
Exponentiated Exponential distribution, that we discussed and determined in
the prevous sections. As mentioned previously this relation will facilitate
our work in determing other expressions.
\begin{theorem}
\label{Thm CDF EHyp}
Let $S_{n} \sim \text{hypoexp}(\overrightarrow{\alpha})$ with coefficients $A_{i}$ defined in Equation \ref{EqAi}, and let $Y \sim \text{EHypoexp}(\overrightarrow{\alpha}, k)$, where $\overrightarrow{\alpha} = (\alpha_{1}, \alpha_{2}, ..., \alpha_{n})$ and $k \in \mathbb{N}^{\ast}$. Then the CDF of $Y$ is
\begin{equation}
F_{Y}(t) = \sum_{\overrightarrow{g}_{i} \in E_{k}} B_{i} F_{N_{i}}(t)
\label{CDF EHyp}
\end{equation}

where $E_{k}=\{\overrightarrow{g}_{i}=(g_{1,i},g_{2,i},\ldots,g_{n,i})\ |\ 0\leq g_{j,i}\leq k,\ \sum_{j=1}^{n}g_{j,i}=k\}$.
 and 
\begin{equation*}
B_{i} = \binom{k}{g_{1,i},\ldots,g_{n,i}} \prod_{j=1}^{n} A_{j}^{g_{j,i}}
\end{equation*}
and denoted by the coefficient of the Exponentiated Hypoexponential
distribution, and $F_{N_{i}}(t)$ is the CDF of $N_{i}\sim MEE(%
\overrightarrow{\alpha },\overrightarrow{g}_{i},k)$.

\begin{proof}
Let $S_{n}\sim \text{hypoexp}(\overrightarrow{\alpha})$ and $\alpha_i > 0$, $i = 1, 2, \ldots, n$. From Eq. \ref{Hyp CDF}, the CDF of $S_{n}$ is given by $F_{S_{n}}(t) = \sum_{i=1}^{n}A_{i}F_{X_{i}}(t)$, where $X_{i}\sim \text{Exp}(\alpha_{i})$ and $A_{i} = \prod_{j=1, j\neq i}^{n}\left(\frac{\alpha_{j}}{\alpha_{j}-\alpha_{i}}\right)$. 
By exponentiating the distribution $S_{n}$ mentioned in Theorem \ref{CDF ExpoRV}, we get $F_{Y}(t) = \left(\sum_{i=1}^{n}A_{i}F_{X_{i}}(t)\right)^{k}$. Using the multinomial expansion, we have:
\begin{eqnarray*}
F_{Y}(t) &=& \sum_{\overrightarrow{g}_{i}\in E_{k}}\binom{k}{g_{1,i},g_{2,i},\ldots,g_{n,i}}\left(A_{1}F_{X_{1}}(t)\right)^{g_{1,i}}\left(A_{2}F_{X_{2}}(t)\right)^{g_{2,i}} \ldots \left(A_{n}F_{X_{n}}(t)\right)^{g_{n,i}} \\
&=& \sum_{\overrightarrow{g}_{i}\in E_{k}}\binom{k}{g_{1,i},g_{2,i},\ldots,g_{n,i}}(A_{1}^{g_{1,i}}A_{2}^{g_{2,i}} \ldots A_{n}^{g_{n,i}})(F_{X_{1}}^{g_{1,i}}(t) F_{X_{2}}^{g_{2,i}}(t) \ldots F_{X_{n}}^{g_{n,i}}(t))
\end{eqnarray*}

where $E_{k}=\{\overrightarrow{g}_{i}=(g_{1,i},g_{2,i},\ldots,g_{n,i})\ |\ 0\leq g_{j,i}\leq k,\ \sum_{j=1}^{n}g_{j,i}=k\}$. But from Corollary \ref{thcdfMEE}, $F_{X_{1}}^{g_{1,i}}F_{X_{2}}^{g_{2,i}}\ldots F_{X_{n}}^{g_{n,i}}$ is the CDF of $N_{i}$, where $N_{i}\sim \text{MEE}(\overrightarrow{\alpha},\overrightarrow{g_{i}})$. Then we can write the CDF of $Y$ as

\begin{equation*}
F_{Y}(t) = \sum_{\overrightarrow{g}_{i}\in E_{k}}B_{i}F_{N_{i}}(t)
\end{equation*}

where $B_{i} = \binom{k}{g_{1,i},\ldots,g_{n,i}}(A_{1}^{g_{1,i}}A_{2}^{g_{2,i}} \ldots A_{n}^{g_{n,i}})$.
\end{proof}
\end{theorem}

In the following theorem, we determine a simplified expression of PDF of the
Exponentiated Hypoexponential distribution of different parameters.

\begin{theorem}
\label{Thm PDF EHyp}Let $Y\sim EHypo\exp \left( n,\overrightarrow{\alpha }%
,k\right) ,\overrightarrow{\alpha }=\left( \alpha _{1},\alpha
_{2},...,\alpha _{n}\right) $ for $1\leq i\leq n$. Then the PDF of $Y$ is 
\begin{equation}
f_{Y}(t)=\sum_{\overrightarrow{g}_{i}\in E_{k}}B_{i}f_{N_{i}}(t)
\label{PDF EHyp}
\end{equation}

\begin{proposition}
Let $B_{i}$ be the coefficient of the Exponentiated Hypoexponential distribution, defined as $B_{i}=\binom{k}{g_{1,i},g_{2,i},\ldots,g_{n,i}} (A_{1}^{g_{1,i}} A_{2}^{g_{2,i}} \ldots A_{n}^{g_{n,i}})$, where $A_{i} = \prod_{j=1, j\neq i}^{n} \left(\frac{\alpha_{j}}{\alpha_{j}-\alpha_{i}}\right)$.
Let $E_{k} = \left\{\overrightarrow{g}_{i} = (g_{1,i}, g_{2,i}, \ldots, g_{n,i}) \middle| 0 \leq g_{j,i} \leq k, \ \sum_{j=1}^{n} g_{j,i} = k\right\}$, and let $N_{i} \sim \text{MEE}(\overrightarrow{\alpha}, \overrightarrow{g_{i}})$.
Then, $\sum_{\overrightarrow{g}_{i} \in E_{k}} B_{i} = 1$.
\end{proposition}

\begin{proof}
From Theorem \ref{Thm CDF EHyp}, we have $F_{Y}\left( t\right) \
=\sum_{E_{k}}B_{i}F_{N_{i}}(t)$ and when we differentiate both sides we
obtain knowing 
\begin{equation*}
f_{Y}(t)=F_{Y}^{\prime }\left( t\right) \ =\frac{d\left( \sum_{%
\overrightarrow{g}_{i}\in E_{k}}B_{i}F_{N_{i}}(t)\right) }{dt}=\sum_{%
\overrightarrow{g}_{i}\in E_{k}}B_{i}\frac{dF_{N_{i}}(t)}{dt}=\sum_{%
\overrightarrow{g}_{i}\in E_{k}}B_{i}f_{N_{i}}(t).
\end{equation*}
\end{proof}
\end{theorem}
\begin{proposition}
\label{sum equal 1}
Let $B_{i}$ be the coefficient of the Exponentiated Hypoexponential distribution. Then $\sum_{\overrightarrow{g}_{i} \in E_{k}} B_{i} = 1$, where $E_{k} = \{\overrightarrow{g}_{i} = (g_{1,i}, g_{2,i}, ..., g_{n,i}) \ |\ 0 \leq g_{j,i} \leq k, \ \sum_{j=1}^{n} g_{j,i} = k\}$.
\end{proposition}

\begin{proof}
Let $F_{Y}(t)$ and $F_{N_{i}}(t)$ be the CDFs of $\text{EHypoexp}(n, \overrightarrow{\alpha}, k)$ and $\text{MEE}(\overrightarrow{\alpha} \overrightarrow{g_{i}}, k)$, respectively. The CDF for any random variable $X$, $P(X < t) = 1$ as $t \rightarrow +\infty$. Thus, $\lim_{t \rightarrow +\infty} F_{Y}(t) = 1$ and $\lim_{t \rightarrow +\infty} F_{N_{i}}(t) = 1$. But from Eq. (\ref{CDF EHyp}), we have: 
\begin{align*}
    F_{Y}(t) &= \sum_{\overrightarrow{g}_{i} \in E_{k}} B_{i} F_{N_{i}}(t) \\
    \lim_{t \rightarrow +\infty} F_{Y}(t) &= \lim_{t \rightarrow +\infty} \sum_{\overrightarrow{g}_{i} \in E_{k}} B_{i} F_{N_{i}}(t) \\
    &= \sum_{\overrightarrow{g}_{i} \in E_{k}} B_{i} \lim_{t \rightarrow +\infty} F_{N_{i}}(t).
\end{align*}
Hence, $\sum_{\overrightarrow{g}_{i} \in E_{k}} B_{i} = 1$.
\end{proof}
\begin{theorem}
\label{Thm Rel EHyp}Let $S_{n}\sim Hypoexp(\overrightarrow{\alpha })\ $where 
$n\in 
\mathbb{N}
^{\ast }$ and $\alpha _{i}>0,$ $i=1,2,...,n$ and let $Y$ be the
Exponentiated Hypoexponential distribution with exponent parameter $k\in 
\mathbb{N}
^{\ast }$ denoted as $Y\sim EHypo\exp \left( n,\overrightarrow{\alpha }%
,k\right) ,\overrightarrow{\alpha }=\left( \alpha _{1},\alpha
_{2},...,\alpha _{n}\right) .$ Then the reliability function of $Y$ is
\end{theorem}

\begin{equation}
R_{Y}\left( t\right) =\sum_{\overrightarrow{g}_{i}\in E_{k}}B_{i}R_{N_{i}}(t)
\label{Rela EHyp}
\end{equation}

and hazard function of $Y$ is

\begin{equation}
h_{Y}\left( t\right) =\frac{\sum_{\overrightarrow{g}_{i}\in
E_{k}}B_{i}f_{N_{i}}(t)}{\sum_{\overrightarrow{g}_{i}\in
E_{k}}B_{i}R_{N_{i}}(t)}  \label{Hazard EHyp}
\end{equation}

\begin{proof}
We have $R_{Y}(t)=1-F_{Y}\left( t\right) $, then and by Theorem \ref{Thm CDF
EHyp}, we get $R_{Y}(t)=1-\sum_{\overrightarrow{g}_{i}\in
E_{k}}B_{i}F_{N_{i}}(t)$. Now, since $F_{N_{i}}(t)=1-R_{N_{i}}\left(
t\right) $ and $\sum_{\overrightarrow{g}_{i}\in E_{k}}B_{i}=1$ from
Proposition \ref{sum equal 1}, we obtain that $R_{Y}\left( t\right) =\sum_{%
\overrightarrow{g}_{i}\in E_{k}}B_{i}R_{N_{i}}(t).$

On the other hand, the hazard function of $Y$ is given as $h_{Y}(t)=\frac{%
f_{Y}(t)}{R_{Y}(t)}$. Also from Theorem \ref{PDF EHyp} we have the
expression of $f_{Y}\left( t\right) $ and substituting the expression of $%
R_{Y}(t)$ in Eq. (\ref{Rela EHyp}) we the obtain our result.
\end{proof}

\subsection{Applications}

\begin{example}
Suppose that $X_{1},$ $X_{2}\ $and $X_{3}$ are the task times of $n=3$ Markov chain that are Exponentially distributed as $X_{j}\sim Exp(\alpha_{j}),$ $j=1,2,3$ with $\alpha_{1}=5,$ $\alpha_{2}=4,$ $\alpha_{3}=3$. $S_{3}=X_{1}+X_{2}+X_{3}$ is the proposed total task time of this system which is a Hypoexponential distribution. Now, let $Y$ be the Exponentiated Hypoexponential distribution with exponent parameter $k=3$. From Theorem \ref{PDF EHyp} 
\begin{equation*}
f_{Y}(t) = \sum_{\overrightarrow{g}_{i}\in E_{k}}B_{i}f_{N_{i}}(t).
\end{equation*}

First $E_{3}=\{\overrightarrow{g}_{i}=(g_{1,i},g_{2,i},g_{3,i})/0\leq g_{j,i}\leq 3,\sum_{j=1}^{3}g_{j,i}=3\}=\{\overrightarrow{g}_{1}=(3,0,0),$ $\overrightarrow{g}_{2}=(0,3,0),$ $\overrightarrow{g}_{3}=(0,0,3),$ $\overrightarrow{g}_{4}=(2,1,0),$ $\overrightarrow{g}_{5}=(2,0,1),$ $\overrightarrow{g}_{6}=(1,2,0),$ $\overrightarrow{g}_{7}=(0,2,1),$ $\overrightarrow{g}_{8}=(1,0,2),$ $\overrightarrow{g}_{9}=(0,1,2),$ $\overrightarrow{g}_{10}=(1,1,1)\}$.

Having $A_{i}=\prod\limits_{j=1, j \neq i}^{n} (\frac{\alpha_{j}}{\alpha_{j}-\alpha_{i}}),$ then $A_{1}=6,$ $A_{2}=-15,$ $A_{3}=10,$

and $B_{i}=\binom{3}{g_{1,i},g_{2,i},g_{3,i}} (A_{1}^{g_{1,i}}A_{2}^{g_{2,i}}A_{3}^{g_{3,i}}),$

and the PDF of $N_{i}\sim MEE(\overrightarrow{\alpha},\overrightarrow{g}_{i})$ derived from \ref{thcdfMEE1} is 
\begin{equation*}
f_{N_{i}}(t) = \prod\limits_{j=1}^{3} (1-e^{-\alpha_{j}t})^{g_{i,j}} \times \sum_{j=1}^{3} \frac{g_{i,j}\alpha_{j}e^{-\alpha_{j}t}}{1-e^{-g\alpha_{j}t}},
\end{equation*}
and then the PDF of $Y$ will have the form 
\begin{eqnarray*}
f_{Y}(t) &=& \sum_{\overrightarrow{g}_{i}\in E_{k}} \binom{3}{g_{1,i},g_{2,i},g_{3,i}} (6)^{g_{1,i}} (-15)^{g_{2,i}} (10)^{g_{3,i}} \\
&& \times \prod\limits_{j=1}^{3} (1-e^{-\alpha_{j}t})^{g_{i,j}} \sum_{j=1}^{3} \frac{g_{i,j}\alpha_{j}e^{-\alpha_{j}t}}{1-e^{-g\alpha_{j}t}} \\
&=& 90e^{-15t}(-1+e^{t})^{8}(6+3e^{t}+e^{2t})^{2}.
\end{eqnarray*}
Its figure is shown in Fig. \ref{fig:1}.
\end{example}
\begin{figure}[htbp]
    \centering
\includegraphics[width=0.6\textwidth]{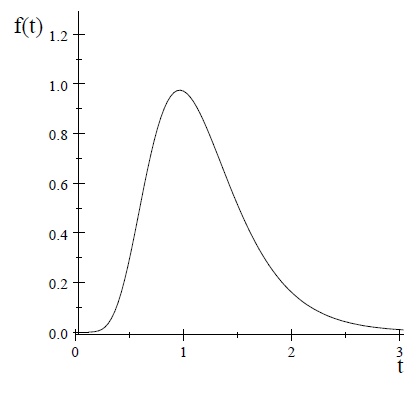}
    \caption{The PDF of the EHyp.}
    \label{fig:1}
\end{figure}

Next, we use Theorem \ref{CDF EHyp} to find the CDF of $Y$ as
\begin{eqnarray*}
F_Y(t) &=& \sum_{\vec{g}_i \in E_k} B_i F_{N_i}(t) \\
&=& \sum_{\vec{g}_i \in E_k} \binom{3}{g_{1,i}, g_{2,i}, g_{3,i}} (6)^{g_{1,i}} (-15)^{g_{2,i}} (10)^{g_{3,i}} \times \prod_{j=1}^{3} (1 - e^{-\lambda_j t})^{\alpha_j} \\
&=& e^{-15t} (-1 + e^{t})^{9} (6 + 3e^{t} + e^{2t})^{3}
\end{eqnarray*}

Note that, if we integrate $f_{Y}(t)$, we get the same result. The figure of 
$F_{Y}(t)$ is represented in Fig \ref{fig:2}.

\begin{figure}[htbp]
    \centering
\includegraphics[width=0.6\textwidth]{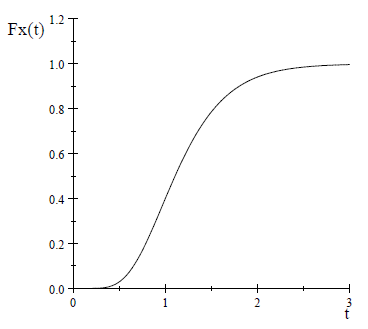}
    \caption{The CDF of the EHyp.}
    \label{fig:2}
\end{figure}

Now, from Theorem \ref{Rela EHyp}%
, we can find reliability function for $Y$ as
\begin{gather*}
R_{Y}(t) = \sum_{\overrightarrow{g}_{i}\in E_{k}}B_{i}R_{N_{i}}(t).
\end{gather*}

However, the expression
of reliability function for Maximum Exponentiated Exponential Distribution
is given in Eq. (\ref{Rela MEE}) as $R_{N}(t) = 1 - \prod_{j=1}^{n} \sum_{i=0}^{g_{j}} \binom{g_{j}}{i} (-1)^{i} e^{-\lambda_{i}(n - g_{j})t_{j}}.$
\begin{eqnarray*}
R_{Y}\left( t\right) &=&\sum_{\overrightarrow{g}_{i}\in
E_{k}}B_{i}R_{N_{i}}(t) \\
&=&1-e^{(-15t)}(-1+e^{t})^{9}(6+3e^{t}+e^{(2t)})^{3}
\end{eqnarray*}

The figure of $R_{Y}(t)$ is represented in Fig. \ref{fig:3}.

\begin{figure}[htbp]
    \centering
\includegraphics[width=0.6\textwidth]{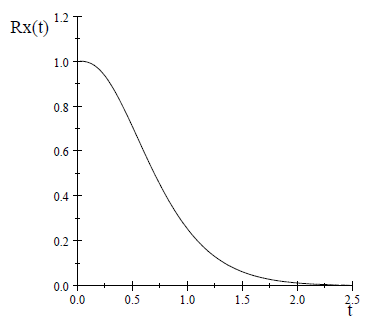}
    \caption{The Reliability of the EHyp.}
    \label{fig:3}
\end{figure}

To understand the concept of $
R_{Y}(2)=0.0588886,$ which means that the probability that the system is
still functioning (survive)\ after 2 years is $0.0588886$ which$\ $is a low
probability. Thus, we can consider that this system will stop after 2 years
or more.

Next, to find the hazard function of $Y,$ use Theorem \ref{Hazard EHyp}. The
hazard's expression is 
\begin{eqnarray*}
h_{Y}(t) &=&\frac{f_{Y}(t)}{R_{Y}(t)} \\
&=&\frac{\sum_{\overrightarrow{g}_{i}\in E_{k}}B_{i}f_{N_{i}}(t)}{\sum_{%
\overrightarrow{g}_{i}\in E_{k}}B_{i}R_{N_{i}}(t)} \\
&=&\frac{%
3(1-6e^{(-5t)}+15e^{(-4t)}-10e^{(-3t)})^{2}(30e^{(-5t)}-60e^{(-4t)}+30e^{(-3t)})%
}{1-e^{(-15t)}(-1+e^{t})^{9}(6+3e^{t}+e^{(2t)})^{3}}
\end{eqnarray*}%

\begin{figure}[htbp]
    \centering
\includegraphics[width=0.6\textwidth]{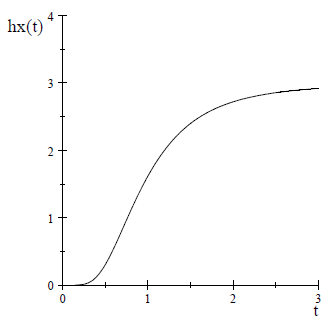}
    \caption{The Hazard of the EHyp.}
    \label{fig:4}
\end{figure}

In Figs. \ref{fig:3} and \ref{fig:4}, we show the curves of reliability and hazard
functions of $Y.$ 
Note that, Fig. \ref{fig:4} shows that the hazard function of $Y$ is not
constant unlike the hazard of the Exponential distribution. The hazard
function of $Y$ is a monotonic increasing function of $t$, then the system
is described to be "wear" out with time, see, \cite{Phd}. In other words,
the system is damaged as a result of normal wearing out or aging with time.
Fig. \ref{fig:4} for $h_{Y}(t)$ reveals that the system takes approximately 10
years to wear out.

\section{Maximum Likelihood Estimation Of Parameters In Exponentiated Hypoexponential Distribution}

\addtolength{\textheight}{.5in}%

Let $X_{1},X_{2},...,X_{N}$ be a random sample from the $EHypo\left( 
\overrightarrow{\alpha },k\right) $ distribution with unknown parameters $%
\alpha _{1},\alpha _{2},...,\alpha _{n},k$. We consider the maximum
likelihood estimation providing the maximum likelihood estimators (MLEs) $%
\widehat{\alpha }_{1},\widehat{\alpha }_{2},...,\widehat{\alpha }_{n},%
\widehat{k}$ for the parameters $\alpha _{1},\alpha _{2},...,\alpha _{n},k$.

Let us recall that the MLEs have some statistical desirable properties
(under regularity conditions) as the sufficiency, invariance, consistency,
efficiency and asymptotic normality. Using the observed information matrix,
asymptotic confidence interval for $\widehat{\alpha }_{1},\widehat{\alpha }%
_{2},...,\widehat{\alpha }_{n},$ and $\widehat{k}$ can be constructed.

Let $Y\sim EHyp\left( \overrightarrow{\alpha },k\right) $, $\alpha _{i}>0,$ $%
i=1,2,...,n,$ and $N$ is a size of a random sample. Then the maximum
likelihood estimators of $\alpha _{1},\alpha _{2},...,\alpha _{n},$ and $k$
is denoted respectively by $\widehat{\alpha }_{1},\widehat{\alpha }_{2},...,%
\widehat{\alpha }_{n}$ and $\widehat{k}$ and has the form:%
\begin{equation*}
\widehat{k}=\frac{-N}{\sum\limits_{j=1}^{N}\log \left( \sum\limits_{i=1}^{n}A_{i}F_{X_{i}}(x_{j})\right) }
\end{equation*}

and $\widehat{a}_{p}$ verify the following implicit equation:

\begin{equation*}
\begin{split}
& \sum_{j=1}^{N} \frac{\sum_{i=1}^{n} \left( \left( \sum_{v=1,v\neq q}^{n} \frac{1}{\alpha _{v}-\alpha _{q}} \right) A_{q} f_{X}(x_{i}) + A_{i} \left( (1-\alpha _{q}x_{j}) e^{-\alpha _{q}x_{j}} \right) \right) }{\sum_{i=1}^{n} A_{i} f_{X}(x_{i})} \\
& + (k-1) \sum_{j=1}^{N} \frac{\sum_{i=1}^{n} \left( \left( \sum_{v=1,v\neq q}^{n} \frac{1}{\alpha _{v}-\alpha _{q}} \right) A_{q} F_{X}(x_{i}) + A_{i} (x_{j}e^{-\alpha _{p}x_{j}}) \right) }{\sum_{i=1}^{n} A_{i} F_{X}(x_{i})} \equiv 0
\end{split}
\end{equation*}

In finding the estimators, the first thing we will do is write the
probability density function as a function of $k$.%
\begin{eqnarray*}
f_{Y}\left( x_{i},\alpha _{i},k\right) &=&k\left( \sum\limits_{i=1}^{n}A_{i}f_{X_{i}}(x_{i})\right) \left( \sum\limits_{i=1}^{n}A_{i}F_{X_{i}}(x_{i})\right) ^{k-1} \\
&=&k\left( \sum\limits_{i=1}^{n}\left( \prod\limits_{j=1,j\neq i}^{n}(\frac{\alpha _{j}}{\alpha _{j}-\alpha _{i}})\right) \left( \alpha _{i}e^{-\alpha _{i}x_{i}}\right) \right) \\
&&\times \left( \sum\limits_{i=1}^{n}\left( \prod\limits_{j=1,j\neq i}^{n}(\frac{\alpha _{j}}{\alpha _{j}-\alpha _{i}})\right) \left( 1-e^{-\alpha _{i}x_{i}}\right) \right) ^{k-1}
\end{eqnarray*}
for $k>0$. Now, that makes the likelihood function:%
\begin{eqnarray*}
L\left( k,\overrightarrow{\alpha }\right) &=&\prod\limits_{j=1}^{N}k\left(
\sum\limits_{i=1}^{n}A_{i}f_{X_{i}}(x_{j})\right) \left(
\sum\limits_{i=1}^{n}A_{i}F_{X_{i}}(x_{j})\right) ^{k-1} \\
&=&k^{N}\left(
\prod\limits_{j=1}^{N}\sum\limits_{i=1}^{n}A_{i}f_{X_{i}}(x_{j})\right)
\left(
\prod\limits_{j=1}^{N}\sum\limits_{i=1}^{n}A_{i}F_{X_{i}}(x_{j})\right)
^{k-1} \\
&=&k^{N}\left( \prod\limits_{j=1}^{N}\sum\limits_{i=1}^{n}\left(
\prod\limits_{j=1,j\neq i}^{n}\left(\frac{\alpha _{j}}{\alpha _{j}-\alpha _{i}}\right)\right) \left( \alpha _{i}e^{-\alpha _{i}x_{i}}\right) \right) \\
&&\times \left( \prod\limits_{j=1}^{N}\left( \prod\limits_{j=1,j\neq
i}^{n}\left(\frac{\alpha _{j}}{\alpha _{j}-\alpha _{i}}\right)\right) \left(
1-e^{-\alpha _{i}x_{i}}\right) \right) ^{k-1}
\end{eqnarray*}

Where $i=1,2,...,n$ ~~and ~$j=1,2,...,N.~$and therefore the log of the
likelihood function:%
\begin{eqnarray*}
\log \left( L\left( k,\overrightarrow{\alpha }\right) \right) &=&\log \left(
k^{N}\left(
\prod\limits_{j=1}^{N}\sum\limits_{i=1}^{n}A_{i}f_{X_{i}}(x_{j})\right)
\left(
\prod\limits_{j=1}^{N}\sum\limits_{i=1}^{n}A_{i}F_{X_{i}}(x_{j})\right)
^{k-1}\right) \\
&=&N\log \left( k\right) +\log \left(
\prod\limits_{j=1}^{N}\sum\limits_{i=1}^{n}A_{i}f_{X_{i}}(x_{j})\right) \\
&&+\log \left( \left(
\prod\limits_{j=1}^{N}\sum\limits_{i=1}^{n}A_{i}F_{X_{i}}(x_{j})\right)
^{k-1}\right) \\
&=&N\log \left( k\right) +\sum\limits_{j=1}^{N}\log \left(
\sum\limits_{i=1}^{n}\left( \prod\limits_{j=1,j\neq i}^{n}\left(\frac{\alpha
_{j}}{\alpha _{j}-\alpha _{i}}\right)\right) \left( \alpha _{i}e^{-\alpha
_{i}x_{i}}\right) \right) \\
&&+(k-1)\left( \sum\limits_{j=1}^{N}\log \left( \left(
\prod\limits_{j=1,j\neq i}^{n}\left(\frac{\alpha _{j}}{\alpha _{j}-\alpha _{i}}%
\right)\right) \left( 1-e^{-\alpha _{i}x_{i}}\right) \right) \right)
\end{eqnarray*}

Now, upon taking the partial derivative of $\log \left( L\left( k,\alpha
_{i}\right) \right) $ with respect to the parameters $k$ and then setting it
to $0$.
\begin{equation*}
\frac{d}{dk}\log \left( L\left( k,\alpha _{i}\right) \right) =0
\end{equation*}
\begin{equation*}
\frac{d}{dk}\left( n\log \left( k\right) +\sum\limits_{j=1}^{n}\log \left(
\sum\limits_{i=1}^{n}A_{i}f_{X_{i}}(x_{j})\right) +(k-1)\left(
\sum\limits_{j=1}^{n}\log \left(
\sum\limits_{i=1}^{n}A_{i}F_{X_{i}}(x_{j})\right) \right) \right) \equiv 0
\end{equation*}
\begin{equation*}
\begin{array}{c}
\frac{d}{dk}(N\log \left( k\right) )+\frac{d}{dk}\left(
\sum\limits_{j=1}^{N}\log \left(
\sum\limits_{i=1}^{n}A_{i}f_{X_{i}}(x_{j})\right) \right) \\
+\frac{d}{dk}\left( (k-1)\left( \sum\limits_{j=1}^{N}\log \left(
\sum\limits_{i=1}^{n}A_{i}F_{X_{i}}(x_{j})\right) \right) \right) \equiv 0
\end{array}
\end{equation*}
\begin{equation*}
\frac{N}{k}+\sum\limits_{j=1}^{N}\log \left(
\sum\limits_{i=1}^{n}A_{i}F_{X_{i}}(x_{j})\right) \equiv 0
\end{equation*}
Now, multiplying through by $k$, and distributing the summation, we get
\begin{equation*}
N+k\left( \sum\limits_{j=1}^{N}\log \left(
\sum\limits_{i=1}^{n}A_{i}F_{X_{i}}(x_{j})\right) \right) \equiv 0
\end{equation*}
Now, solving for $k$, and putting its estimate, we have shown that the MLE
of $k$ is
\begin{equation*}
\widehat{k}=\frac{-N}{\sum\limits_{j=1}^{N}\log \left(
\sum\limits_{i=1}^{n}A_{i}F_{X_{i}}(x_{j})\right) }
\end{equation*}
Again, by taking the partial derivative of the log likelihood with respect
to $\alpha _{q}$, and setting it $0$, leaving us with:%
\begin{equation*}
\frac{\partial }{\partial \alpha _{q}}\log \left( L\left( k,\alpha _{i}\right) \right) =0
\end{equation*}

\begin{equation*}
\frac{\partial }{\partial \alpha _{q}}\left( \begin{array}{c}
N\log \left( k\right) +\sum\limits_{j=1}^{N}\log \left( \sum\limits_{i=1}^{n}\left( \prod\limits_{j=1,j\neq i}^{n}\frac{\alpha _{j}}{\alpha _{j}-\alpha _{i}}\right) \left( \alpha _{i}e^{-\alpha _{i}x_{j}}\right) \right) \\
+(k-1)\left( \sum\limits_{j=1}^{N}\log \left( \prod\limits_{j=1}^{N}\left( \prod\limits_{j=1,j\neq i}^{n}\frac{\alpha _{j}}{\alpha _{j}-\alpha _{i}}\right) \left( 1-e^{-\alpha _{i}x_{j}}\right) \right) \right)
\end{array}%
\right) \equiv 0
\end{equation*}
\begin{equation*}
\begin{aligned}
&\frac{\partial }{\partial \alpha _{p}}(N\log \left( k\right) )+\frac{\partial }{\partial \alpha _{p}}\left( \sum\limits_{j=1}^{N}\log \left( \sum\limits_{i=1}^{n}\left( \prod\limits_{j=1,j\neq i}^{n}\frac{\alpha _{j}}{\alpha _{j}-\alpha _{i}}\right) \left( \alpha _{i}e^{-\alpha _{i}x_{j}}\right) \right) \right) \\
&+\frac{\partial }{\partial \alpha _{p}}\left( (k-1)\left( \sum\limits_{j=1}^{N}\log \left( \sum\limits_{i=1}^{n}\left( \prod\limits_{j=1,j\neq i}^{n}(\frac{\alpha _{j}}{\alpha _{j}-\alpha _{i}})\right) \left( 1-e^{-\alpha _{i}x_{j}}\right) \right) \right) \right) \equiv 0
\end{aligned}
\end{equation*}
\begin{equation*}
\begin{array}{c}
\sum\limits_{j=1}^{N}\frac{\sum\limits_{i=1}^{n}\left(
\prod\limits_{j=1,j\neq q}^{n}\frac{\alpha _{j}}{\alpha _{j}-\alpha _{q}}%
\right) \left( \sum\limits_{v=1,v\neq q}^{n}\frac{1}{\alpha _{v}-\alpha _{q}%
}\right) \left( \alpha _{i}e^{-\alpha _{i}x_{j}}\right) +\left(
\prod\limits_{j=1,j\neq q}^{n}\frac{\alpha _{j}}{\alpha _{j}-\alpha _{q}}%
\right) \left( \left( 1-\alpha _{q}x_{j}\right) e^{-\alpha _{q}x_{j}}\right) 
}{\sum\limits_{i=1}^{n}\left( \prod\limits_{j=1,j\neq i}^{n}\frac{\alpha
_{j}}{\alpha _{j}-\alpha _{i}}\right) \left( \alpha _{i}e^{-\alpha
_{i}x_{i}}\right) }+(k-1) \\ 
\times \sum\limits_{j=1}^{N}\frac{\sum\limits_{i=1}^{n}\left( \left(
\left( \prod\limits_{j=1,j\neq q}^{n}\frac{\alpha _{j}}{\alpha _{j}-\alpha
_{q}}\right) \left( \sum\limits_{v=1,v\neq q}^{n}\frac{1}{\alpha
_{v}-\alpha _{q}}\right) \right) \left( 1-e^{-\alpha _{i}x_{j}}\right)
+\left( \prod\limits_{j=1,j\neq q}^{n}\frac{\alpha _{j}}{\alpha _{j}-\alpha
_{q}}\right) \left( x_{j}e^{-\alpha _{p}x_{j}}\right) \right) }{%
\sum\limits_{i=1}^{n}\left( \prod\limits_{j=1,j\neq i}^{n}(\frac{\alpha
_{j}}{\alpha _{j}-\alpha _{i}})\right) \left( 1-e^{-\alpha _{i}x_{j}}\right) 
}\equiv 0%
\end{array}
\end{equation*}
\begin{equation*}
\begin{array}{c}
\sum\limits_{j=1}^{N}\frac{\sum\limits_{i=1}^{n}\left( \left(
\sum\limits_{v=1,v\neq q}^{n}\frac{1}{\alpha _{v}-\alpha _{q}}\right)
A_{q}f_{X}\left( x_{i}\right) +A_{i}\left( \left( 1-\alpha _{q}x_{j}\right)
e^{-\alpha _{q}x_{j}}\right) \right) }{\sum\limits_{i=1}^{n}A_{i}f_{X}%
\left( x_{i}\right) } \\ 
+(k-1)\sum\limits_{j=1}^{N}\frac{\sum\limits_{i=1}^{n}\left( \left(
\sum\limits_{v=1,v\neq q}^{n}\frac{1}{\alpha _{v}-\alpha _{q}}\right)
A_{q}F_{X}\left( x_{i}\right) +A_{i}\left( x_{j}e^{-\alpha _{p}x_{j}}\right)
\right) }{\sum\limits_{i=1}^{n}A_{i}F_{X}\left( x_{i}\right) }\equiv 0
\end{array}
\end{equation*}
Some numerical iterative methods, \ such as Newton-Raphson method can be
used since this equation is not solvable analytically. The solution can be
approximate numerically by using software such as MAPLE, R, and MATHEMATICA.
Here we worked with MATHEMATICA, see, Wolfram (1999).

\subsection{ILLUSTRATIVE REAL DATA EXAMPLES}

In this subsection, we analyze real data sets to show that the $EHypo\exp
\left( \overrightarrow{\alpha },k\right) $ distribution can be a better
model than other existing distributions.

Let $Y\sim EHypo\exp \left( \overrightarrow{\alpha },k\right) $, where $%
\overrightarrow{\alpha }$ $=\left( \alpha _{1},\alpha _{2}\right) $ is the
EHypoexponential with exponent parameter $k$. The data set contains $n=128$
measures on the remission times in months of bladder cancer patients. It is
extracted from Lee and Wang \cite{20003}:

\begin{tabular}{l}
$%
0.08,2.09,3.48,4.87,6.94,8.66,13.11,23.63,0.20,2.23,3.52,4.98,6.97,9.02,13.29, 
$ \\ 
$%
0.40,2.26,3.57,5.06,7.09,9.22,13.80,25.74,0.50,2.46,3.64,5.09,7.26,9.47,14.24, 
$ \\ 
$%
25.82,0.51,2.54,0.08,2.09,3.48,4.87,6.94,8.66,13.11,23.63,0.20,2.23,3.52,4.98, 
$ \\ 
$%
6.97,9.02,13.29,0.40,2.26,3.57,5.06,7.09,9.22,13.80,25.74,0.50,2.46,3.64,5.09, 
$ \\ 
$%
7.26,9.47,14.24,25.82,0.51,2.54,3.70,5.17,7.28,9.74,14.76,26.31,0.81,2.62,3.82, 
$ \\ 
$%
5.32,7.32,10.06,14.77,32.15,2.64,3.88,5.32,7.39,10.34,14.83,34.26,0.90,2.69, 
$ \\ 
$%
4.18,5.34,7.59,10.66,15.96,36.66,1.05,2.69,4.23,5.41,7.62,10.75,16.62,43.01, 
$ \\ 
$1.19,2.75,4.26,5.41,7.63,17.12,46.12,1.26,2.83,4.33,7.66,11.25,17.14,79.05,$
\\ 
$%
1.35,2.87,5.62,7.87,11.64,17.36,1.40,3.02,4.34,5.71,7.93,11.79,18.10,1.46,4.40, 
$ \\ 
$%
5.85,8.26,11.98,19.13,1.76,3.25,4.50,6.25,8.37,12.02,2.02,3.31,4.51,6.54,8.53, 
$ \\ 
$%
12.03,20.28,2.02,3.36,6.76,12.07,21.73,2.07,3.36,6.93,8.65,12.63,22.69,5.49. 
$%
\end{tabular}

The probability density function of $Y$ is:%
\begin{eqnarray*}
f_{Y}(x_{i},\alpha_{i},k) &=& k \left( \sum_{i=1}^{n} \left( \prod_{j=1,j\neq i}^{n} \left( \frac{\alpha_{j}}{\alpha_{j}-\alpha_{i}} \right) \right) \left( \alpha_{i}e^{-\alpha_{i}x_{i}} \right) \right) \\
&& \times \left( \sum_{i=1}^{n} \left( \prod_{j=1,j\neq i}^{n} \left( \frac{\alpha_{j}}{\alpha_{j}-\alpha_{i}} \right) \right) \left( 1-e^{-\alpha_{i}x_{i}} \right) \right)^{k-1}
\end{eqnarray*}
for $-\infty <x<\infty ,$ with parameters $k~$and $\alpha _{i}.$ Thus, the
likelihood function is:%
\begin{eqnarray*}
L(x_{i},\alpha_{i},k) &=& \prod_{j=1}^{N} k \left( \frac{e^{-x_{j}\alpha_{2}}\alpha_{1}\alpha_{2}}{\alpha_{1}-\alpha_{2}} - \frac{e^{-x_{j}\alpha_{1}}\alpha_{1}\alpha_{2}}{\alpha_{2}-\alpha_{1}} \right) \\
&& \times \left( 1 - \frac{e^{-x_{j}\alpha_{2}}\alpha_{1}}{\alpha_{1}-\alpha_{2}} - \frac{e^{-x_{j}\alpha_{1}}\alpha_{2}}{\alpha_{2}-\alpha_{1}} \right)^{k-1} \\
&=& k^{N} \left( \prod_{j=1}^{N} \left( \frac{e^{-x_{j}\alpha_{2}}\alpha_{1}\alpha_{2}}{\alpha_{1}-\alpha_{2}} - \frac{e^{-x_{j}\alpha_{1}}\alpha_{1}\alpha_{2}}{\alpha_{2}-\alpha_{1}} \right) \right) \\
&& \times \left( \prod_{j=1}^{N} \left( 1 - \frac{e^{-x_{j}\alpha_{2}}\alpha_{1}}{\alpha_{1}-\alpha_{2}} - \frac{e^{-x_{j}\alpha_{1}}\alpha_{2}}{\alpha_{2}-\alpha_{1}} \right)^{k-1} \right)
\end{eqnarray*}
Where $i=1,2,...,n$ ~~and ~$j=1,2,...,N.$

It can be shown, upon maximizing the likelihood function with respect to $k$
,$\alpha _{1}$ and $\alpha _{2}$, that the maximum likelihood estimator of $%
k $ is:%
\begin{equation*}
\widehat{k} = \frac{-N}{\sum\limits_{j=1}^{N} \log \left(\sum\limits_{i=1}^{n} A_{i}F_{X_{i}}(x_{j})\right)}
\end{equation*}%
and $\alpha _{1}$ and $\alpha _{2}$ verify the following implicit Eq.%
\begin{equation*}
\begin{array}{c}
\sum\limits_{j=1}^{N}\frac{\sum\limits_{i=1}^{n}\left( \left(
\sum\limits_{v=1,v\neq q}^{n}\frac{1}{\alpha _{v}-\alpha _{q}}\right)
A_{q}f_{X}\left( x_{i}\right) +A_{i}\left( \left( 1-\alpha _{q}x_{j}\right)
e^{-\alpha _{q}x_{j}}\right) \right) }{\sum\limits_{i=1}^{n}A_{i}f_{X}%
\left( x_{i}\right) } \\ 
+(k-1)\sum\limits_{j=1}^{N}\frac{\sum\limits_{i=1}^{n}\left( \left(
\sum\limits_{v=1,v\neq q}^{n}\frac{1}{\alpha _{v}-\alpha _{q}}\right)
A_{q}F_{X}\left( x_{i}\right) +A_{i}\left( x_{j}e^{-\alpha _{p}x_{j}}\right)
\right) }{\sum\limits_{i=1}^{n}A_{i}F_{X}\left( x_{i}\right) }\equiv 0%
\end{array}%
\end{equation*}
Based on the given sample, and using MATHEMATICA the maximum likelihood
estimate of $k$ ,$\alpha _{1}$ and $\alpha _{2}$ are:%
\begin{eqnarray*}
\widehat{k} &=&0.51541 \\
\widehat{\alpha }_{1} &=&0.105418 \\
\widehat{\alpha }_{2} &=&0.718769
\end{eqnarray*}%

For the data set, we compare the fitted distributions using the criteria: -2log(L), Akaike information criterion (AIC), Akaike information criterion corrected (AICC), Bayesian information criterion (BIC), Anderson-Darling ($A^{\ast}$), and Cramer-von Mises ($W^{\ast}$). Let us be accurate that $\log(L)$ is the log-likelihood taken with the estimated values, AIC=-2log(L)+2c, AICC = AIC + $\frac{2c(c+1)}{v-c-1}$, and BIC = -2log(L) + clog(v), where c denotes the number of estimated parameters and v denotes the sample size. The best-fitted distribution corresponds to lower -2log(L), AIC, AICC, BIC, $A^{\ast}$, and $W^{\ast}$. We see in Tables 1 and 2 that the $EHyp\exp (\left\{ a_{1},a_{2}\right\},k)$ distribution has the smallest -2log(L), AIC, AICC, BIC, $A^{\ast}$, and $W^{\ast}$ for the two data sets (except for the data set where the Hypexp distribution has the smallest BIC), indicating that it is a serious competitor to the other considered distributions. Also, graphically, the behaviors of our model, Figure \ref{fig:2}, show the fitted densities superimposed on the histogram of the given data sets.
\newline
\begin{table}[htbp]
    \centering
    \begin{tabular}{|l|l|l|l|l|l|}
\hline
$\text{Model}$ & $\text{Estimate of parameters}$ & $-2log(L)$ & $AIC$ & $%
AICC $ & $BIC$ \\ \hline
$Hyp\exp $ & $%
\begin{tabular}{r}
$\widehat{\alpha }_{1}=0.114536$ \\ 
$\widehat{\alpha }_{2}=1.57547$%
\end{tabular}%
$ & $826.09$ & $1082.09$ & $822.058$ & $914.813$ \\ \hline
$EHyp\exp $ & 
\begin{tabular}{l}
$\widehat{\alpha }_{1}=0.105418$ \\ 
$\widehat{\alpha }_{2}=0.718769$ \\ 
$\widehat{k}=0.51541$%
\end{tabular}
& $823.33$ & $1079.33$ & $817.234$ & $963.952$ \\ \hline
\end{tabular}
\end{table}

\begin{table}[htbp]
    \centering
    \begin{tabular}{|l|l|l|}
        \hline
        $\text{Model}$ & $A^{\ast }$ & $W^{\ast }$ \\ \hline
        $Hyp$ & $0.541467$ & $0.0889889$ \\ \hline
        $EHyp$ & $0.455972$ & $0.0767364$ \\ \hline
    \end{tabular}
    \caption{Results for Model Parameters}
    \label{tab:model_results}
\end{table}

\begin{figure}[htbp]
    \centering
    \includegraphics[width=0.6\textwidth]{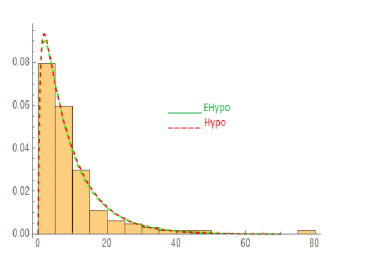}
    \caption{The Fitted Densities Superimposed on the Histogram of Data Sets.}
    \label{fig:5}
\end{figure}
\section{Conclusion And Open Problems}
\label{sec:conc}
\bigskip
\begin{center}
    {\large\bf CHARACTERIZATION \& ESTIMATION}
\end{center}

In this article we investigated the problem of characterizing the PDF, some statistical functions and maximum likelihood estimation of parameters of the EHypoexponential random variables that have applications in many scientific areas.

We derived expressions for PDF, CDF, and some statistical parameters of Maximum Exponentiated Exponential random variable with different parameters that are later used to obtain our results. Moreover, we determined closed expressions for PDF, CDF, some statistical parameters and examined the Maximum Likelihood Estimation of parameters of EHypoexponential Distribution.

\end{document}